\documentclass[conference]{IEEEtran}

\usepackage{amsmath,amssymb,graphicx,cite,algorithm,algcompatible,array,booktabs,multirow,float,amsthm,dsfont}
\usepackage{relsize}

\usepackage{caption}

\newtheorem{proposition}{Proposition}

\def\qeds{\qed\par\medskip}
\def\qedsf{\vskip-8mm\qeds}

\DeclareMathOperator*{\argmax}{argmax}

\newcommand{\mathleft}{\@fleqntrue\@mathmargin0pt}

\makeatletter
\def\ps@IEEEtitlepagestyle{%
  \def\@oddfoot{\mycopyrightnotice}%
  \def\@evenfoot{}%
}
\def\mycopyrightnotice{%
  {\footnotesize 978-1-5386-4505-5/18/\$31.00~\copyright~2018 IEEE\hfill}
  \gdef\mycopyrightnotice{}
}
\makeatother

\begin{document}
    \title{Privacy-preserving smart meter control strategy including energy storage losses}
    
    \author{\IEEEauthorblockN{Ramana R. Avula, Tobias J. Oechtering and Daniel M${\text{\r{a}}}$nsson}
        \IEEEauthorblockA{School of Electrical Engineering and Computer Science\\
            KTH Royal Institute of Technology, Stockholm, Sweden}}
    
    \maketitle
    
    \begin{abstract}
        Privacy-preserving smart meter control strategies proposed in the literature so far make some ideal assumptions such as instantaneous control without delay, lossless energy storage systems etc. In this paper, we present a one-step-ahead predictive control strategy using Bayesian risk to measure and control privacy leakage with an energy storage system. The controller estimates energy state using a three-circuit energy storage model to account for steady-state energy losses. With numerical experiments, the controller is evaluated with real household consumption data using a state-of-the-art adversarial algorithm. Results show that the state estimation of the energy storage system significantly affects the controller's performance. The results also show that the privacy leakage can be effectively reduced using an energy storage system but at the expense of energy loss. 
    \end{abstract}
    \bigskip
    
    %
    \IEEEpeerreviewmaketitle
    
    \section{Introduction}
    \label{sec:intro}
    
    A smart grid (SG) is a next-generation energy network with capabilities to improve grid reliability and efficiency of power generation and distribution with smooth integration of renewable energy sources. In this automated network, a smart meter (SM) is a crucial component which measures the energy consumption of the user and transmits the readings to the utility provider at regular intervals of time. This raises privacy concerns \cite{mcdaniel2009security} since high-resolution readings can allow anyone who has access to this data to infer about consumer's behavior. Since its introduction in \cite{hart1992nonintrusive}, non-intrusive load monitoring (NILM) techniques are known to be quite effective in disaggregating the smart meter readings and thereby detecting the states of most of the general types of household appliances \cite{zoha2012non}. A comparative study was done in \cite{zoha2012non}, which shows that the existing state of the art NILM algorithms are capable of achieving detection accuracy up to 99\% for certain appliance types, which is quite concerning in the privacy context. 
    
    Addressing this issue, several privacy-preserving techniques have been proposed in the literature, which are surveyed in \cite{asghar2017smart,giaconi2018privacy}. Secure communication and cryptographic approaches \cite{lemay2007unified,jawurek2011smart,rial2011privacy} may succeed in preventing the unauthorized third party access, but they would fail to protect the consumer privacy from a greedy authorized or compromised utility provider. A promising physical layer privacy approach is load signature moderation (LSM), where an energy storage system (ESS) is used to moderate the consumer's load profile in order to hide appliances' usage information. LSM using rechargeable battery has previously been investigated in \cite{kalogridis2010privacy,varodayan2011smart, tan2013increasing,backes2014differentially,li2016privacy,chin2017privacy} to obtain optimal control strategy under different privacy settings. However, all the works so far make some ideal assumptions such as instantaneous control without delay, lossless ESS etc. These idealized strategies may provide theoretic performance bounds but the feasibility of such strategies in practical situations must be further investigated. 
    \begin{figure}[t] 
        \begin{minipage}[b]{1.0\linewidth}
            \centering
              \vspace{0.3cm}
            \centerline{\includegraphics[width=1\linewidth]{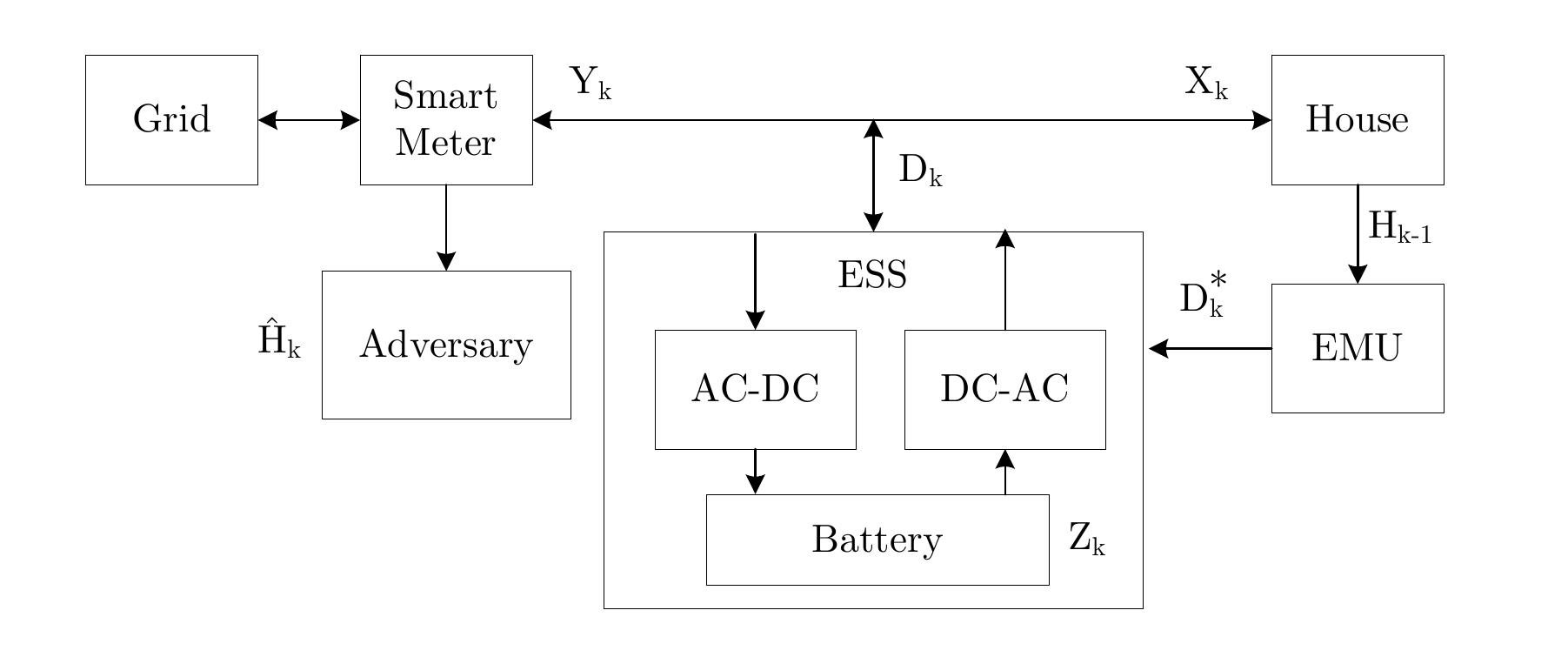}}
              \vspace{0.3cm}
        \end{minipage}
        \caption{Schematic of the proposed smart metering system where the energy management unit controls privacy leakage to an adversary by using energy storage system with a model describing its losses and one-step-ahead predictive control. }
        \label{fig:sys_mod}
    \end{figure}
    In this paper, we present a one-step-ahead predictive control scheme modeled in a PO-MDP framework using an ESS. Similar to \cite{li2016privacy}, we use a privacy metric based on Bayesian risk. The overview of the proposed system is shown in Fig.~\ref{fig:sys_mod}. In this work, we restrict our analysis to the electrochemical battery as an energy storage technology. Nonetheless, the same approach can be followed for other storage technologies by modeling them as their equivalent electrical circuits \cite{pham2017physical}. For a battery system, we present a model describing its losses in power conversion, losses due to internal resistance and self-dissipation. To the best of our knowledge, this is the first work to consider the non-idealities in ESSs in the context of smart meter privacy. 
    
    The rest of the paper is organized as follows. In Section \ref{sec:bat_model}, we present a model for ESS considering the steady state energy losses. We also present the charge and discharge bounds of ESS and also quantify the energy loss associated with a discrete control action. In Section \ref{sec:control}, we present an overview of the system along with the control strategy. In Section \ref{sec:num}, we evaluate the performance of the controller with real household data using a state-of-the-art NILM algorithm. Lastly, we conclude the paper in Section \ref{sec:conc}.
    \section{Energy Storage system Model}
    \label{sec:bat_model}
    Since batteries stores the energy as chemical potential in their electrodes, it can only be interfaced with a DC (Direct Current) system. Hence, power converters are needed to integrate the battery with an AC (Alternating Current) system. The battery along with the power converters form the ESS. In this work, we model the ESS using three simple electrical circuits as shown in Fig.~\ref{fig:bat_mod} to account for the steady state energy losses. Even though several other processes of the ESS such as capacity fade, increase in internal resistance, temperature dependence etc., can also be considered, as a first step, we restrict our focus to the steady state energy losses. In the following, we present our analysis of the three-circuit model in more detail.
    \begin{figure}[h]
        \begin{minipage}[b]{1\linewidth}
            \centering
            \centerline{\includegraphics[width=\textwidth]{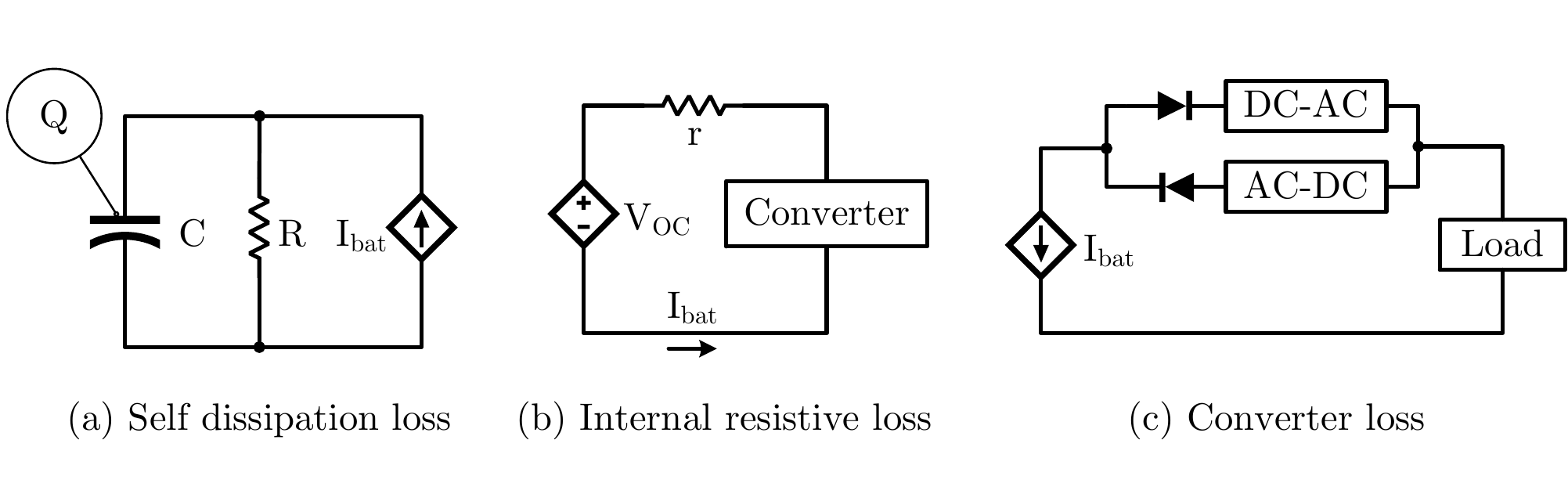}}
        \end{minipage}
        \caption{Three-circuit energy storage system model}
        \label{fig:bat_mod}
    \end{figure}
    \subsection{Losses due to self-dissipation}
    Self-dissipation occurs even if the ESS is not connected to any load. Similar to \cite{chen2006accurate}, we model this phenomenon using an RC circuit as shown in Fig.~\ref{fig:bat_mod}(a).
    The capacitor C holds the charge content, Q of the battery and dissipates through a parallel resistor R. In this circuit, the power converter and load are together represented as a current source controlled by the current flowing in the second circuit shown in Fig.~\ref{fig:bat_mod}(b). Given the self-discharge rate $\gamma$ and a constant current  $\text{I}_{\text{bat}}$ flowing into the battery, the charge content of the battery is updated as
    \begin{gather} \label{eq:self_diss}
    {\text{Q}}_{t+\Delta t} = (1-\gamma)\cdot{\text{Q}}_{t}+\beta\cdot{{\text{I}}_{\text{bat}}}    
    \end{gather}
    \noindent where,
    \begin{gather}
    \gamma = 1- \exp\Big(\frac{-\Delta T}{\text{RC}}\Big); \qquad \beta = \frac{-\gamma \Delta T}{\log(1-\gamma)}
    \end{gather}
    \subsection{Losses due to internal resistance}
    Similar to \cite{chen2006accurate}, the losses that occur in battery during its charging and discharging operations are modeled using a series resistor $r$, as shown in Fig.~\ref{fig:bat_mod}(b). The open circuit voltage of the battery, $\text{V}_{\text{OC}}$ is represented as a voltage source controlled by the capacitor charge Q in the first circuit shown in Fig.~\ref{fig:bat_mod}(a). 
    For an input power of $\text{P}$ from the power converter, the current flowing into the battery is given as
    \begin{equation} \label{eq:int_res_loss}
    {\text{I}}_{\text{bat}} = \frac{\sqrt{\text{V}_{\text{OC}}^{2}+4\text{rP}}-\text{V}_{\text{OC}}}{2\text{r}} 
    \end{equation}
    \subsection{Losses due to power converters}
    As shown in Fig.~\ref{fig:bat_mod}(c), we model power converters as elements with a constant efficiency factor within their operating region. For an input power of $\text{D}$ from the load, the power at the battery terminals can be written as
    \begin{flalign} \label{eq:pow_loss}
    \text{P} &= \text{D}\cdot \left( \eta_{c}\mathds{1}\{\text{D} \ge 0\} +  \eta_{d}^{-1}\mathds{1}\{\text{D} < 0\} \right) \nonumber \\
        &= \text{D}\cdot \delta(\text{D})
    \end{flalign}
    where $\eta_{c}, \eta_{d}$ are the efficiency factors of AC-DC and DC-AC converters respectively and $\mathds{1}\{A\}$ is equal to 1 if A is true, and 0 otherwise. $\delta(\text{D})$ is the common factor for both operations.
    \subsection{Three-circuit ESS model}
    Integrating the three circuits by combining (\ref{eq:self_diss}), (\ref{eq:int_res_loss}) and (\ref{eq:pow_loss}), the controller updates the energy state of the battery evolving over time using the equation given as 
    \begin{equation} \label{eq:upd_eq}
    {\text{Z}}_{t+\Delta t} = (1-\gamma){\text{Z}}_{t}+\frac{\beta \text{V}_{\text{OC}}}{2\text{r}}\left(\sqrt{\text{V}_{\text{OC}}^{2}+4\text{r}\text{D}_{t} \delta_t} - \text{V}_{\text{OC}}\right)
    \end{equation}
    where $\gamma$, $r$ are time-invariant parameters, $\beta$ depends on the time step $\Delta t$ and $\delta_t$ depends on the control variable $\text{D}_{t}$. By limiting the battery current to $\text{I}_{\text{max}}$, from (\ref{eq:int_res_loss}), we have the control space limited as
    \begin{flalign}
    \text{D}_{\text{max}} &= \frac{1}{4\text{r}\eta_d}\left({(\text{V}_{\text{OC}} + 2\text{rI}_{\text{max}} )^2 - \text{V}_{\text{OC}}^2}\right) \label{eq:con2}\\[0.2cm]
    \text{D}_{\text{min}} &= \frac{\eta_c}{4\text{r}}\left({\text{V}_{\text{OC}}^2 - (\text{V}_{\text{OC}} - 2\text{rI}_{\text{max}} )^2}\right) \label{eq:con1}
    \end{flalign}
    Due to the finite energy capacity of the battery $\text{Z}_{max}$, from (\ref{eq:upd_eq}) we have the following constraints on $\text{D}_{t}$:
    \begin{flalign}
    \text{D}_{t,max} &= \frac{\text{V}_{\text{OC}}^2}{4\text{r}\eta_d}\left(\left(\frac{2\text{r}\big[\text{Z}_{max}-(1-\gamma){\text{Z}}_{t}\big]}{\beta\text{V}_{\text{OC}}^2}+1\right)^2 - 1\right) \label{eq:con3} \\[0.2cm]
    \text{D}_{t,min} &= \frac{\eta_c\text{V}_{\text{OC}}^2}{4\text{r}}\left(\left(\bigg[\frac{-(1-\gamma)2\text{r}{\text{Z}}_{t}}{\beta\text{V}_{\text{OC}}^2}+1\bigg]^{+}\right)^2 - 1\right)\label{eq:con4}
    \end{flalign}
    where $[x]^{+}$ is equal to $x$ if $x \ge 0$, and 0 otherwise. In comparision to (\ref{eq:upd_eq}), the energy state of an ideal lossless battery evolves over time as
    \begin{equation} \label{eq:upd_eq_ideal}
    {\text{Z}}_{t+\Delta t,ideal} = {\text{Z}}_{t}+\text{D}_{t}\cdot\Delta t
    \end{equation}
    Using (\ref{eq:upd_eq}) and (\ref{eq:upd_eq_ideal}), the energy loss associated with a discrete control action can be given as
    \begin{equation} \label{eq:engloss}
    \mathcal{E}_{loss}({\text{Z}}_{t},\text{D}_{t}) = {\text{Z}}_{t}+\text{D}_{t}\cdot\Delta t - {\text{Z}}_{t+\Delta t}
    \end{equation}
    \section{System overview and Control strategy}
    \label{sec:control}
    The proposed smart metering system uses an ESS for load signature moderation. The ESS can be placed either in series or in parallel configurations, in between the SM and house as shown in Fig.~\ref{fig:bat_placements}. Both these configurations have been used in the literature for SM privacy. Under ideal assumptions, the two configurations are equivalent. However, considering the energy losses, we have the following proposition.
    \begin{proposition} \label{prop:1}
        The average energy loss in the parallel configuration is strictly less than that of series configuration.
    \end{proposition} 
    \begin{proof}
        Let ${\text{X}}$ be the average energy demand by the house and ${\text{Y}}$ be the average energy request from the grid. Assuming that the energy from the battery is not discharged back into the grid, only $({\text{Y}} - {\text{X}})$ flows through the ESS components in parallel configuration, however, the total energy ${\text{Y}}$ from the grid flows through the ESS components in series configuration, leading to higher energy losses.   
    \end{proof}
    \begin{figure}[t]    
        \begin{minipage}[b]{1.0\linewidth}
            \centering
            \centerline{\includegraphics[width=0.95\linewidth]{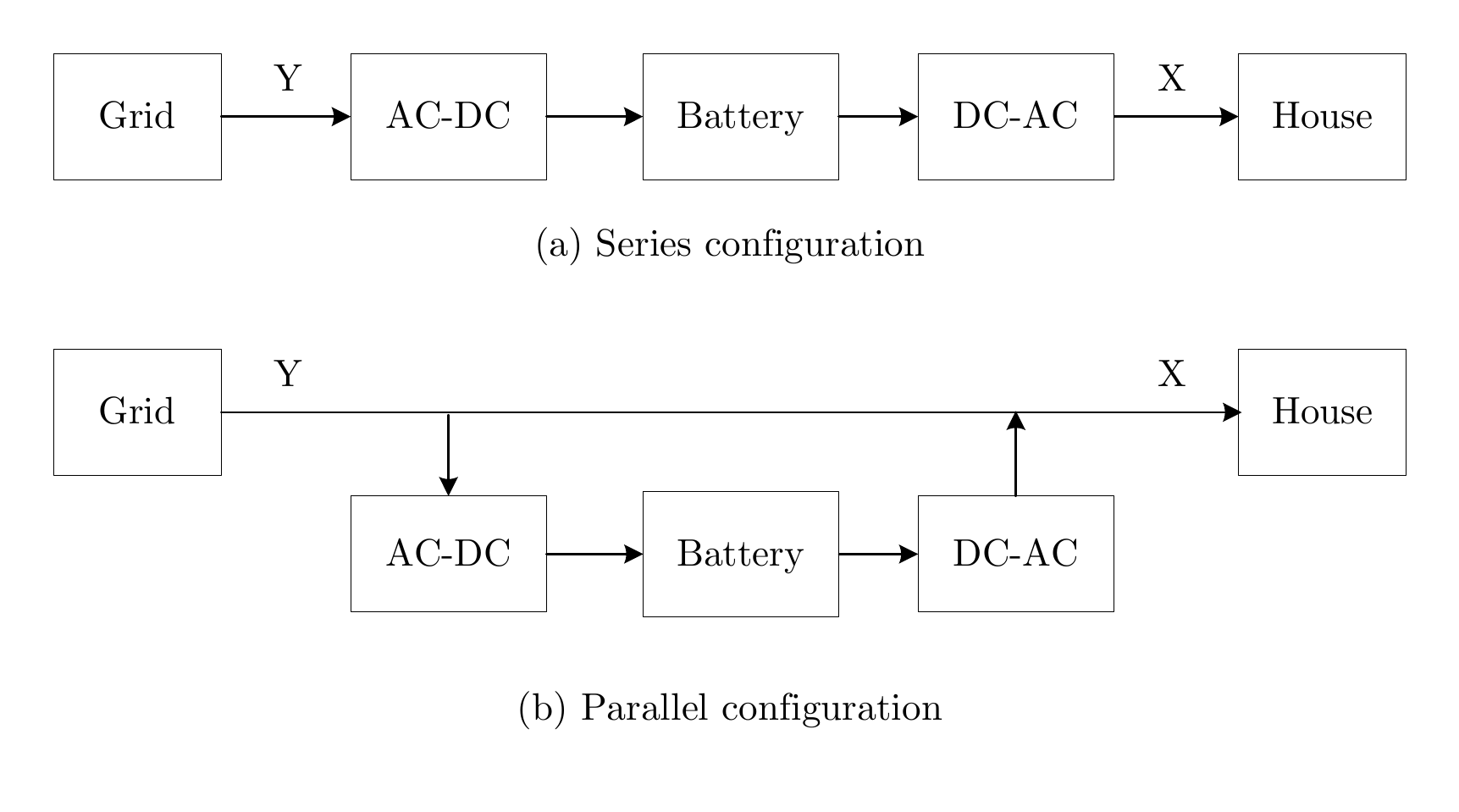}}
        \end{minipage} 
        %
        \caption{Placement of ESS between smart meter and house.}
        \label{fig:bat_placements}
    \end{figure}
    In order to reduce the energy losses, we consider a system with ESS in the parallel configuration as shown in Fig.~\ref{fig:sys_mod}. The discrete time system is controlled for every time slot $k$ within a finite time horizon $\{1,2,\dots,N\}$. Each time slot $k$ is of a fixed time duration $T$. Let $e$ and $q$ be the resolution of energy and power measurements respectively. In the following analysis, we use the capital letters to denote random variables, their realizations by the lower-case letters and the range space by calligraphic letters.
    
    For each time slot $k$, $X_{k}$ denotes the aggregate power drawn by all the appliances in the house and is defined on $\cal{X}$$ = \{0,q,2q,$ $\dots,x_{max}\}$. $Z_{k}$ defined on $\cal{Z}$$ = \{0,e,2e,\dots,z_{max}\}$ denotes the energy available in the battery. The power drawn by the ESS is denoted as $D_{k}$ and it is the control variable which is defined on $\mathcal{D} = \{-d_{min},\dots,-q,0,q,\dots,d_{max}\}$, where $d_{min}$ and $d_{max}$ are the maximum discharge and charge power of the ESS respectively, given in (\ref{eq:con2}) and (\ref{eq:con1}). $D_{k}^{*}$ defined on $\cal{D}$ denotes the desired battery power consumption scheduled by the energy management unit (EMU). In the presence of an ESS, the SM records the aggregate power demands of consumer and ESS. In this work, we allow the energy from the battery to be discharged to the grid resulting in negative values of SM measurements. It is represented by the random variable $Y_{k} = X_{k}+D_{k}$ which is defined on $\mathcal{Y} = \{-d_{min},\dots,-q,0,q,\dots,x_{max}+d_{max}\}$. $H_{k}$ defined on $\cal{H}$ denotes the $n$-ary joint hypothesis of all the appliances in the house and $\hat{H}_{k}$ defined on $\cal{H}$ denotes the hypothesis detected by the adversary having access to the consumer's statistical and real-time data as well as the control strategy employed by the EMU. 
    
    \subsection{Bayesian risk}
    Similar to \cite{li2016privacy}, we use detection-theoretic approach by formulating the smart meter privacy problem into an adversarial Bayesian hypothesis testing where an adversary having access to the consumer's statistical and real-time data tries to make a guess on the hypothesis state using a decision strategy. In the Bayesian formulation, each of the hypothesis test outcomes is assigned a cost and the decision strategy that minimizes the average decision-making cost will be employed by the adversary \cite{varshney2012distributed}. The average cost or \emph{Bayesian risk} function, $\mathcal{R}$, is given as 
    \begin{equation} \label{eq:risk}
    \mathcal{R}_{k}=\sum_{i,j \in \mathcal{H}^2} C_{i,j}\cdot P(\hat{H}_{k}=i\mid H_{k}=j)\cdot P(H_{k}=j) 
    \end{equation} 
    where $C_{i,j}$ is the cost of deciding $\hat{H}_{k}=i$ when $H_{k}=j$ is true. By setting the cost of a correct decision to zero and the cost of an error to unity, the risk function gives the average error probability of an adversarial detection strategy. 
    
    In this work, the \emph{accumulated minimum Bayesian risk} (AMBR) is chosen as a privacy metric, which is given as 
    \begin{equation} \label{eq:risk_acc}
    \text{AMBR} = \sum_{k=1}^{N}\mathcal{R}_{k}^{*}
    \end{equation}
    where $\mathcal{R}_{k}^{*} = \min \{\mathcal{R}_{k}\}$. The AMBR is a good choice for measuring privacy due to its operational meaning. It explicitly characterizes the best possible detection performance achievable by any adversary. 
    
    \subsection{Control strategy}
    Ideally, the controller uses all the information available until time $k$ (denoted as $\mathcal{I}_{k}$) to choose an action $d_{k+1}$. However, as described in \cite{krishnamurthy2016partially}, since $\mathcal{I}_{k}$ is increasing in dimension with $k$, its sufficient statistic given by the posterior distribution of the Markov chain $H_{k}$ conditioned on $\mathcal{I}_{k}$ (denoted as $\pi_{k}$) is used instead of $\mathcal{I}_{k}$ to choose the action $d_{k+1}$. For a given initial battery state $z_{0}$, this posterior distribution forms a \emph{information state} or \emph{belief state} at time $k$, given as
    \begin{flalign}
    \pi_{k}(i) &= {P}(H_{k}=i\mid\mathcal{I}_{k}) \nonumber \\
    &= {P}(H_{k}=i\mid\pi_{k-1},x_{k})\label{eq:belief_eq}
    \end{flalign}
    where $\mathcal{I}_{k} = \{z_{0},\pi_{0},x_{1},y_{1},\pi_{1},\dots,\pi_{k-1},x_{k}\}$. The control system is modeled as a PO-MDP controlled sensor, as shown in Fig.~\ref{fig:con_mod}, by making the following assumptions:
    \begin{itemize}
        \item The hypothesis of the house $H_{k}$ evolves over time following a first-order Markov chain with a time-invariant transition probability $P_{H_{k}\mid H_{k-1}}$.
        \item The controller observes the Markov chain $H_{k}$ only through a noisy measurement $X_{k}$ made with a time-invariant observation probability $P_{X_{k}\mid H_{k}}$.
        \item The control signal  $D_{k}^{*}$ is generated by the controller using time-dependent control strategy $P_{Y_{k}\mid X_{k-1},Z_{k-1},\Pi_{k-1}}$.\\
    \end{itemize}
    \begin{figure}[t] 
        \begin{minipage}[b]{1\linewidth}
            \centering
            \centerline{\includegraphics[width=0.85\linewidth]{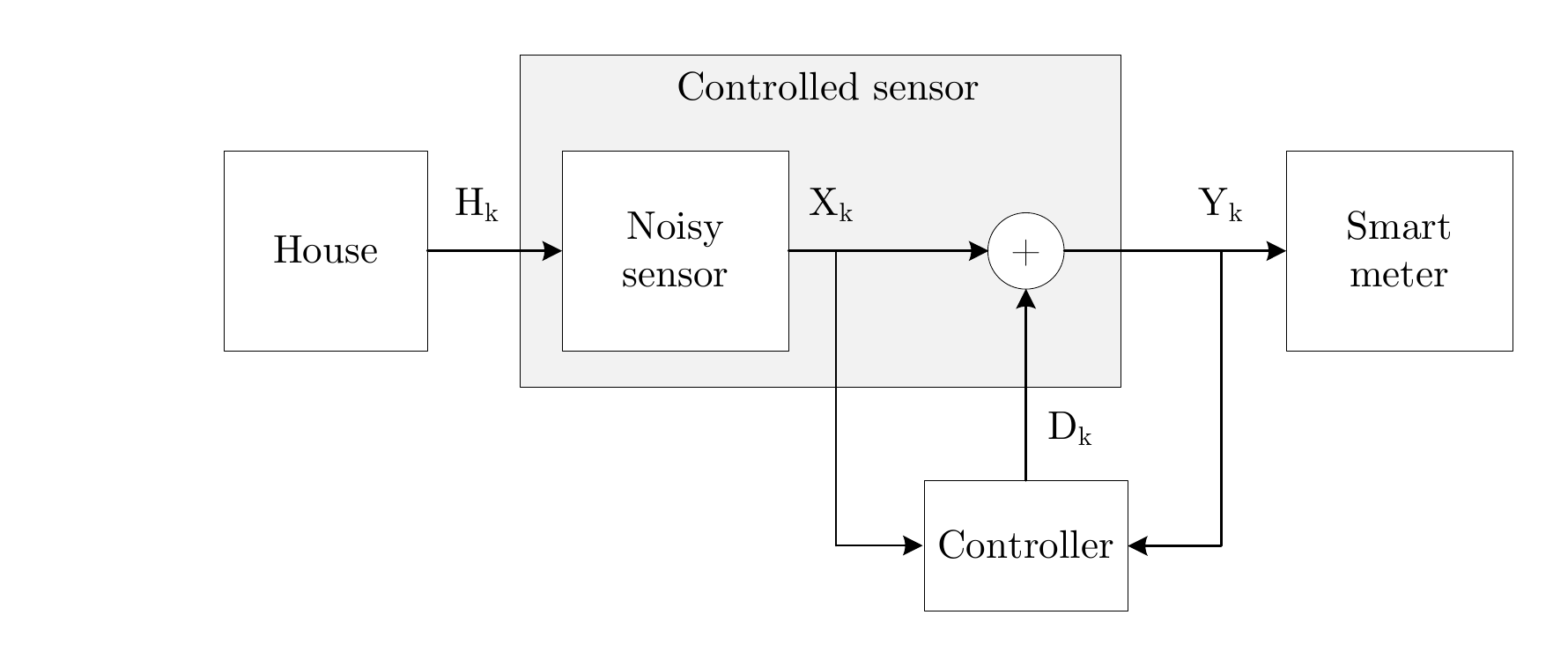}} 
        \end{minipage}
        \caption{EMU modelled as a PO-MDP controlled sensor.}
        \label{fig:con_mod}
        \vspace*{-0.2cm}
    \end{figure}    
    Given the initial energy state, $z_0$, the controller estimates the state of ESS at any time $k$ using the equation
    \begin{flalign}
    z_{k} &= f(z_{k-1},d_k) \label{eq:bat_eq}
    \end{flalign}
    where $d_k = y_k-x_k$ and $f$ is a deterministic function of ESS model given by (\ref{eq:upd_eq}). As described in \cite{varshney2012distributed}, the minimum Bayesian risk function based on our assumptions is given as
    \begin{align*} 
    &\mathcal{R}_{k}^{*}(\pi_{k-1},z_{k-1},\mu_{k}) = \sum_{y \in \mathcal{Y}} \min_{\hat{h} \in \mathcal{H}} \Bigl\{\sum_{g,h,x \in \mathcal{H}^2\times\mathcal{X}}C(\hat{h},h)\cdot \nonumber\\
    & \qquad \qquad \qquad {P}_{Y_{k}\mid X_{k-1},Z_{k-1}}(y\mid x,z) \cdot {P}_{X_{k-1}\mid H_{k-1}}({x\mid g})\cdot \nonumber \\
    & \qquad \qquad \qquad {P}_{H_{k}\mid H_{k-1}}({h\mid g}) \cdot {P}_{H_{k-1}}({g})\Bigr\} 
    \end{align*} 
    For the finite time horizon, the optimal control strategy is the solution to the nonlinear optimization problem with objective function given as
    \begin{flalign} \label{eq:risk_acc1}
    \mu^* = \argmax_{\{\mu_{1},\cdots,\mu_{N}\}}\sum_{k=1}^{N}\mathcal{R}_{k}^{*}(\pi_{k-1},z_{k-1},\mu_{k}) 
    \end{flalign} 
    subject to the constraints given as
    \begin{flalign*}
    P_{Y_{k}\mid Z_{k-1}}(y_k\mid z_{k-1}) = 0 \text{ if } 
    \begin{cases}
    y_{k} < {d}_{k,min} \\
    \qquad \text{ or}  \\
    y_{k} > {d}_{k,max} + x_{max}
    \end{cases}
    \end{flalign*}
    
    The \emph{belief state space} (denoted as $\Pi$) is a $\vert\mathcal{H}\vert-1$ dimensional unit simplex. Solving this optimization problem requires discretization of $\Pi$ in order to get a finite set. The optimization variable in (\ref{eq:risk_acc1}) is of dimension $N\times \vert\mathcal{Y}\vert\times \vert\mathcal{X}\vert$ and solving it in its original form is computationally complex as the dimensionality of the the problem increases with $N$. This can be formulated into a recursive dynamic programming problem as given in the following proposition, the proof of which follows from \cite{krishnamurthy2016partially}.
    \begin{proposition} \label{prop:2}
        For the finite horizon PO-MDP with model given in Section \ref{sec:control}, the optimal control strategy $\mu^{*} = \{\mu_{1}^{*},\mu_{2}^{*},$ $\dots,\mu_{N}^{*}\}$ is the solution to the following backward recursion: Initialize $\mathcal{V}_{N}(\pi,z)$ and then for $k=N-1,\dots,1$ iterate
        \begin{align*}
        & L_{k}(\pi_{k-1},z_{k-1},\mu_{k})= \mathcal{R}_{k}^{*}(\pi_{k-1},z_{k-1},\mu)  ~+  \nonumber \\
        & \qquad \qquad \qquad \qquad \qquad\sum_{x_{k},y_{k} \in \mathcal{X}\times\mathcal{Y}} \mathcal{V}_{k+1}(\pi_{k},z_{k}) \cdot P(x_{k},y_{k})  \nonumber \\
        & \mathcal{V}_{k}(\pi_{k-1},z_{k-1}) = \max_{\mu_{k}}\Bigl\{L_{k}(\pi_{k-1},z_{k-1},\mu_{k})\Bigr\}  \nonumber \\
        & \delta_{k}^{*}(\pi_{k-1},z_{k-1}) = \argmax_{\mu_{k}}\Bigl\{L_{k}(\pi_{k-1},z_{k-1},\mu_{k})\Bigr\}     
        \end{align*}
        \qedsf    
    \end{proposition} 
    
    With the designed optimal strategy  $\mu^{*}$, a real-time PO-MDP controller is implemented as shown in the following algorithm.
    \begin{algorithm} [h]
        \caption{Realtime PO-MDP controller}
        \begin{algorithmic}[1]
            \renewcommand{\algorithmicrequire}{\textbf{Initialisation}:}
            \REQUIRE $\pi_{0},z_{0}$
            \FOR {$k = 1$ to $N$}
            \STATEx \textit{Pre-process} :
            \STATE Choose action $y_{k}^{*} = \mu_{k}^{*}(\pi_{k-1},z_{k-1})$
            \STATEx \textit{ESS control} :
            \IF {($y_{k}^{*} < x_{k} + d_{k,min}$)}
            \STATE Limit $y_{k} = x_{k} + d_{k,min}$
            \ELSIF {($y_{k}^{*} > x_{k} + d_{k,max}$)}
            \STATE Limit $y_{k} = x_{k} + d_{k,max}$
            \ELSE 
            \STATE Allow $y_{k} =y_{k}^{*}$
            \ENDIF
            \STATEx \textit{Post-process} :
            \STATE Update the belief state $\pi_{k} = T(\pi_{k-1},x_{k})$
            \STATE Update the ESS state $z_{k} = f(z_{k-1},y_{k}-x_{k})$
            \ENDFOR
        \end{algorithmic} \label{alg:controller}
    \end{algorithm}
    \section{Numerical Experiments}
    \label{sec:num}
    The simulation experiments to validate our control scheme are implemented in MATLAB using real household consumption data from ECO reference dataset \cite{beckel2014nilm}. The control strategy is obtained by solving the optimization problem using the nonlinear programming solver based on interior point algorithm \cite{byrd2000trust}. We simulate a scenario where the controller is tasked to protect the events of a water kettle every day between 8 AM and 9 AM. The controller chooses an action every minute by observing the real-time appliance consumption data. For this objective, a 12V 100Ah lithium-ion battery is selected, which can sufficiently satisfy the power requirements of the kettle. To simplify the problem, we assume a fixed $\text{V}_{\text{OC}}$ equal to the nominal battery voltage. The parameters used in the simulation are listed in Table \ref{tab:simparam}. The Markov chain probabilities of the PO-MDP control model are estimated from 30 days of labeled training data and listed in the Table \ref{tab:conparam}.   
    \subsection{Visualization of control actions}
    With this setup, the control actions for different initial states of the battery are simulated and are shown in Fig.~\ref{fig:fig_bat_control}. Due to the measurement quantization, switching events are noticed as peaks in the smart meter measurements as shown in Fig.~\ref{fig:fig_sm_readings}. These residual peaks are informative to an adversary operating with high precision measurements. However, the cardinality of the state space increases with measurement precision, which increases the dimensionality of the optimization problem by $\mathcal{O}(n^2)$. Fig.~\ref{fig:fig_bat_soc} shows the evolution of the state of charge (SOC) of the battery due to control actions. It is interesting to notice that without any design objective on the battery state, the control scheme is steering the battery towards the full charge state. This result in the degradation of controller's performance which is discussed in the following. 
    \begin{table}[t]
        \captionsetup{justification=centering}
        \caption{Simulation parameters}
        \label{tab:simparam}
        \centering
        \begin{tabular}{| l | l | c |} 
            \hline
            \textbf{Parameter} & \textbf{Symbol} & \textbf{Value} \\ 
            \hline
            Max. appliance power demand                         & $x_{\text{max}}$ (W)         & 1700  \\ 
            Time slot duration                         & T (s)                     & 60  \\
            Time horizon length                         & N                     & 60 \\
            Power measurement resolution                        & q (W)                 & 500  \\
            Energy measurement resolution                         & e (Wh)                 & 5  \\ 
            Battery nominal voltage                 & $\text{V}_{\text{nom}}$ (V)     & 12  \\ 
            Battery capacity                         & $\text{Q}_{\text{max}}$ (Ah)    & 100  \\ 
            Max. allowed battery charging current                & $\text{I}_{\text{max}}$ (A)     & 80  \\ 
            Max. allowed battery discharging current            & $\text{I}_{\text{min}}$ (A)    & 80  \\ 
            Battery internal resistance                & r ($\Omega$)             & 0.006  \\ 
            Battery self-discharge rate                  & $\gamma$ (\%/month)                & 3  \\ 
            Power converter efficiency                 & $\eta_c, \eta_d$ (\%)             & 95  \\ 
            \hline
            Cardinality of $\mathcal{X}$             & $\vert\mathcal{X}\vert$        & 4  \\ 
            Cardinality of $\mathcal{Y}$             & $\vert\mathcal{Y}\vert$        & 8  \\ 
            Cardinality of $\mathcal{Z}$             & $\vert\mathcal{Z}\vert$        & 241  \\ 
            Cardinality of $\mathcal{H}$             & $\vert\mathcal{H}\vert$        & 2  \\  
            Cardinality of ${\Pi}$                     & $\vert{\Pi}\vert$                & 11  \\  
            Max. allowed battery input power         & $\text{D}_{\text{max}}$         & +2q  \\ 
            Min. allowed battery input power         & $\text{D}_{\text{min}}$         & -2q  \\ 
            \hline 
            ESS model parameter                        & $\beta$        & 0.017  \\ 
            \hline
        \end{tabular}
        \vskip\baselineskip
        \caption{PO-MDP control parameters}
        \label{tab:conparam}
        \centering
        \renewcommand{\arraystretch}{1}
        \begin{tabular}{| l | c |} 
            \hline
            \textbf{Parameter} & \textbf{Value} \\ 
            \hline
            $C_{i,j}$                        &$\begin{bmatrix}0 &1\\1 &0\end{bmatrix}$             \\ [2mm]
            $\pi_{0}$                        &$\begin{bmatrix}0.95 & 0.05\end{bmatrix}^{T}$                    \\ [2mm]
            $P(H_{k}|H_{k-1})$                 &$\begin{bmatrix}0.98 &0.34\\0.02 &0.65\end{bmatrix}$ \\ [2mm]
            $P(X_{k}|H_{k})$                &$\begin{bmatrix}1 &0 &0 &0\\0 &0.17 &0.14 &0.17\end{bmatrix}^{T}$ \\ [2mm]
            \hline
        \end{tabular}
    \end{table}
    
    \subsection{Evaluation of controller's performance}
    To evaluate the performance of the controller, we simulate an adversary using NILM algorithm. In particular, we simulated Weiss' algorithm \cite{weiss2012leveraging} which extracts switching events from the aggregate smart meter data and assigns each event to the appliance with the best match in a signature database. This algorithm is implemented using NILM toolbox developed by \cite{beckel2014nilm}. We use 30 days of labeled training data to create the signature database. Weiss’ algorithm utilizes three-dimensional consumption data (i.e., real, reactive, and distortion powers) in order to match an event signature. We tested its detection performance by injecting the controlled battery current in-phase to the supply voltage resulting in corrupted real power measurements. The accuracy of the adversarial detection is measured using F-score, which is given as
    \begin{equation}
    \text{F-score} = \frac{1}{1+(\text{FN}+\text{FP})/(2\text{TP})} \label{eq:fscore}
    \end{equation}
    where $\text{FN}$, $\text{FP}$ and $\text{TP}$ denote false negative, false postive and true positive respectively. The F-score lies between 0 and 1, where F-score = 0 indicates no detection and F-score = 1 indicates complete detection. 
    \begin{figure}[t]
        \vskip \baselineskip
        \begin{minipage}[b]{1\linewidth}
            \centering
            \centerline{\includegraphics[width=1\linewidth]{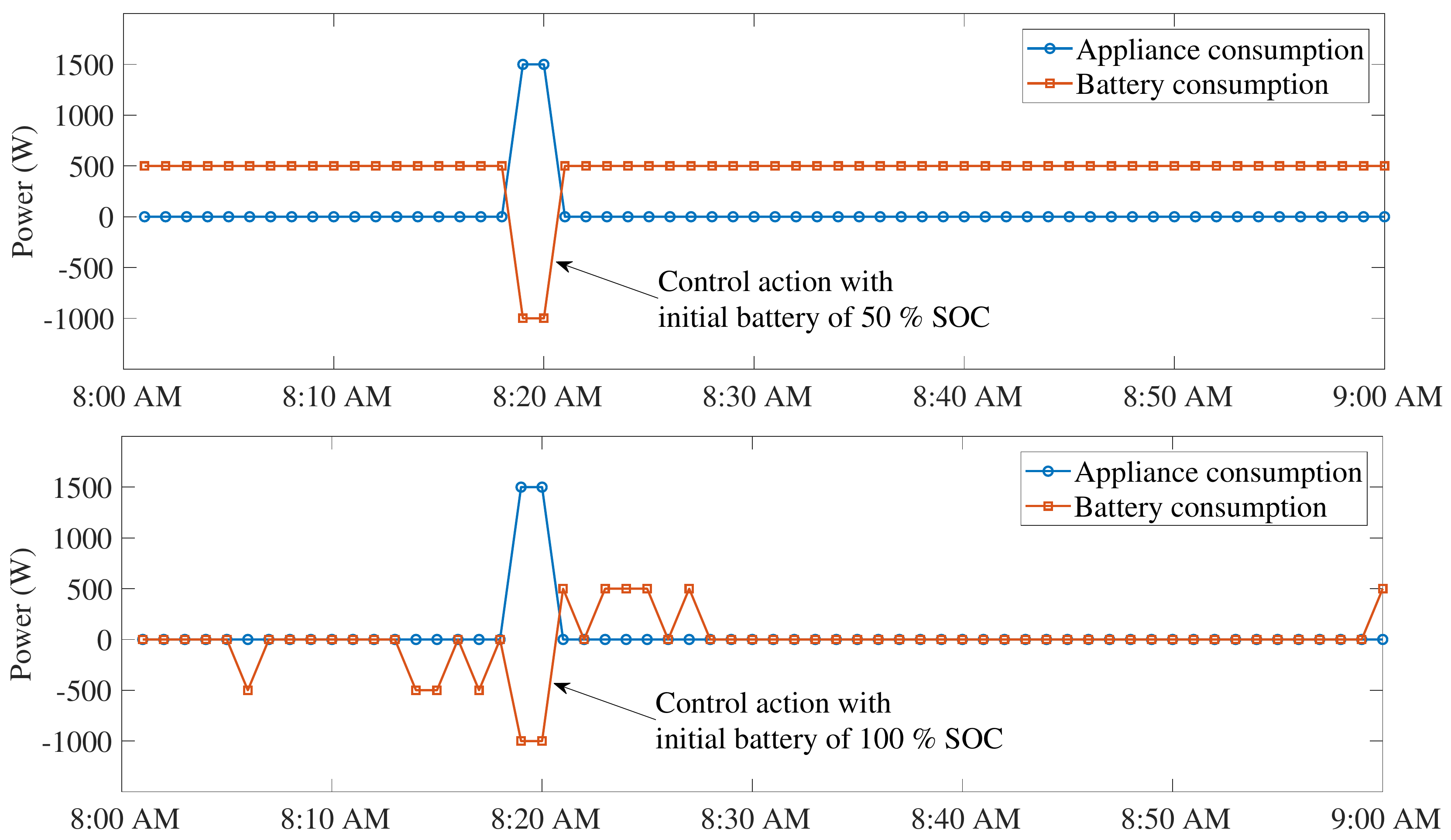}} 
        \end{minipage}
        \caption{Control actions of battery.}
        \label{fig:fig_bat_control}    
        \vskip \baselineskip
        \begin{minipage}[b]{1\linewidth}
            \centering
            \centerline{\includegraphics[width=1\linewidth]{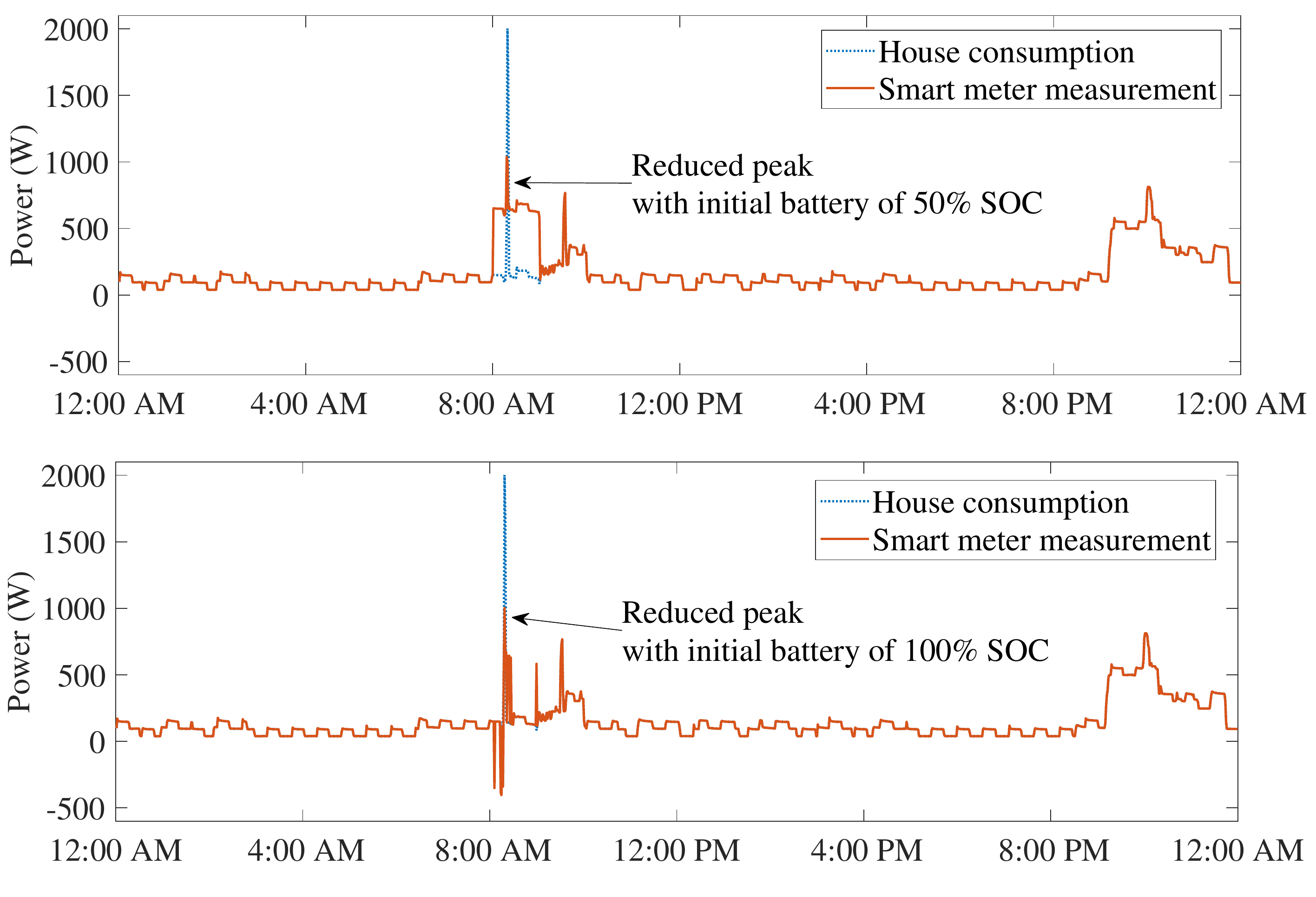}} 
        \end{minipage}
        \caption{Smart meter readings vs household consumption.}
        \label{fig:fig_sm_readings}
        \vskip \baselineskip
        \begin{minipage}[b]{1\linewidth}
            \centering
            \centerline{\includegraphics[width=0.96\linewidth]{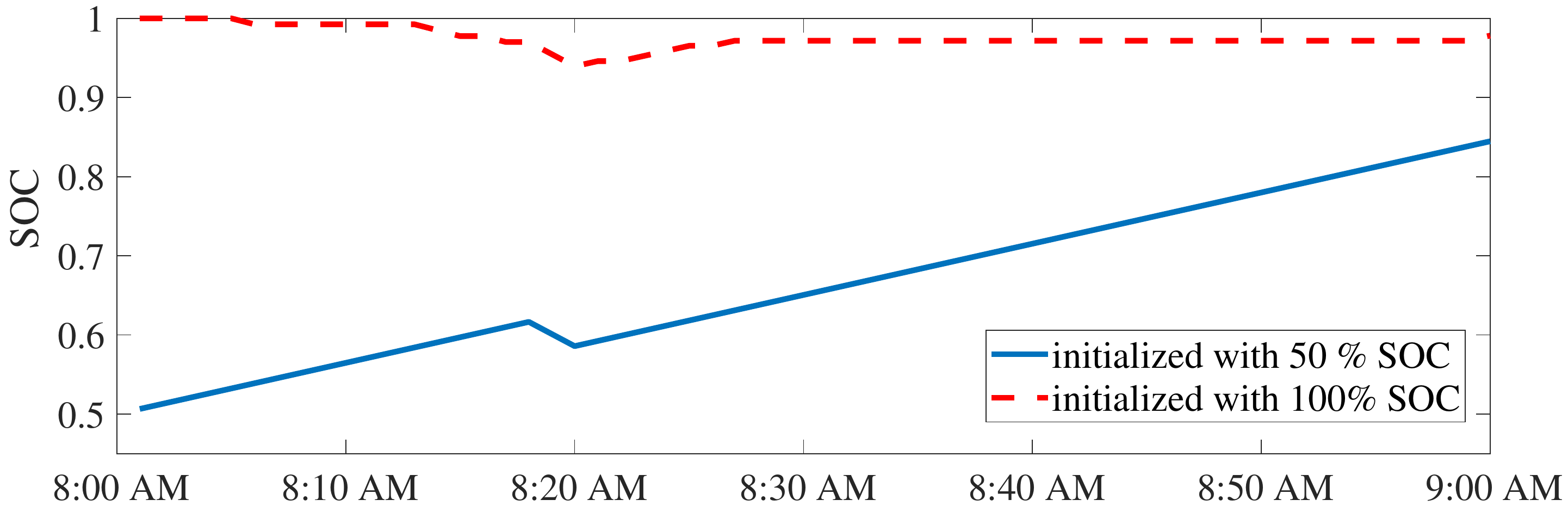}} 
        \end{minipage}
        \caption{Evolution of SOC of battery due to control actions.}
        \label{fig:fig_bat_soc}
    \end{figure}
    
    The Weiss's algorithm is simulated under different test conditions using 30 days of validation data and the obtained average F-scores, energy losses and the AMBR are listed in Table~\ref{tab:conperf}. It can be seen from the results that the AMBR and the F-score are correlated. The test case without a battery resulted in the highest F-score. While using a battery, the controller is able to reduce the F-score significantly. However, it can be observed that the ability of the controller to preserve privacy depends on the initial state of the battery. For this simulation setup, the controller performed better when initialized with a battery of 25\%-50\% SOC compared to full charge. This indicates that if the battery state is steered towards 25\%-50\% SOC by the end of the control time horizon, it would result in better performance for the next control horizon. However, this improved performance is achieved at the cost of increased energy loss. 
    \subsection{ESS model comparision}
    For the simulated battery, Fig.~\ref{fig:batmodelcomp} shows the difference between the \% change in the state of charge of the battery estimated by three-circuit model and an ideal lossless model for different input powers. The model difference is particularly significant at high power levels. Due to very low $\gamma$ for electrochemical batteries, the difference in state estimation is negligible when comparing three-circuit models with and without considering self-dissipation. However, for energy storage systems with high self-dissipation rate such as flywheels, the self-dissipation phenomenon cannot be neglected. 
    \begin{table}[t]
        \caption{Evaluation of controller against NILM algorithm with different initial battery states}
        \label{tab:conperf}
        \centering
        \renewcommand{\arraystretch}{1}
        \begin{tabular}{| c | c | c | c |} 
            \hline
            \textbf{Initial battery SOC (\%)}     & \textbf{F-score}     & \textbf{Energy loss (Wh)}  & \textbf{AMBR}                   \\ 
            \hline
            0                                 & 0.1333             & 40.421                     & 152.57                        \\ 
            \hline
            25                                 & 0.0357             & 36.217                     & 153.51                        \\ 
            \hline
            50                                 & 0.0357             & 36.230                    & 153.51                        \\ 
            \hline
            75                                & 0.2667             & 26.770                    & 153.50                        \\ 
            \hline
            90                                 & 0.4833             & 14.174                    & 153.49                          \\ 
            \hline
            100                                 & 0.6333             & 9.779                     & 148.89                         \\ 
            \hline    
            Without battery                        & 0.7931             & 0                            & 0                                \\ 
            \hline
        \end{tabular}
    \end{table}  
    \begin{figure}[h]
        \begin{minipage}[b]{1\linewidth}
            \centering
            \centerline{\includegraphics[width=0.8\linewidth]{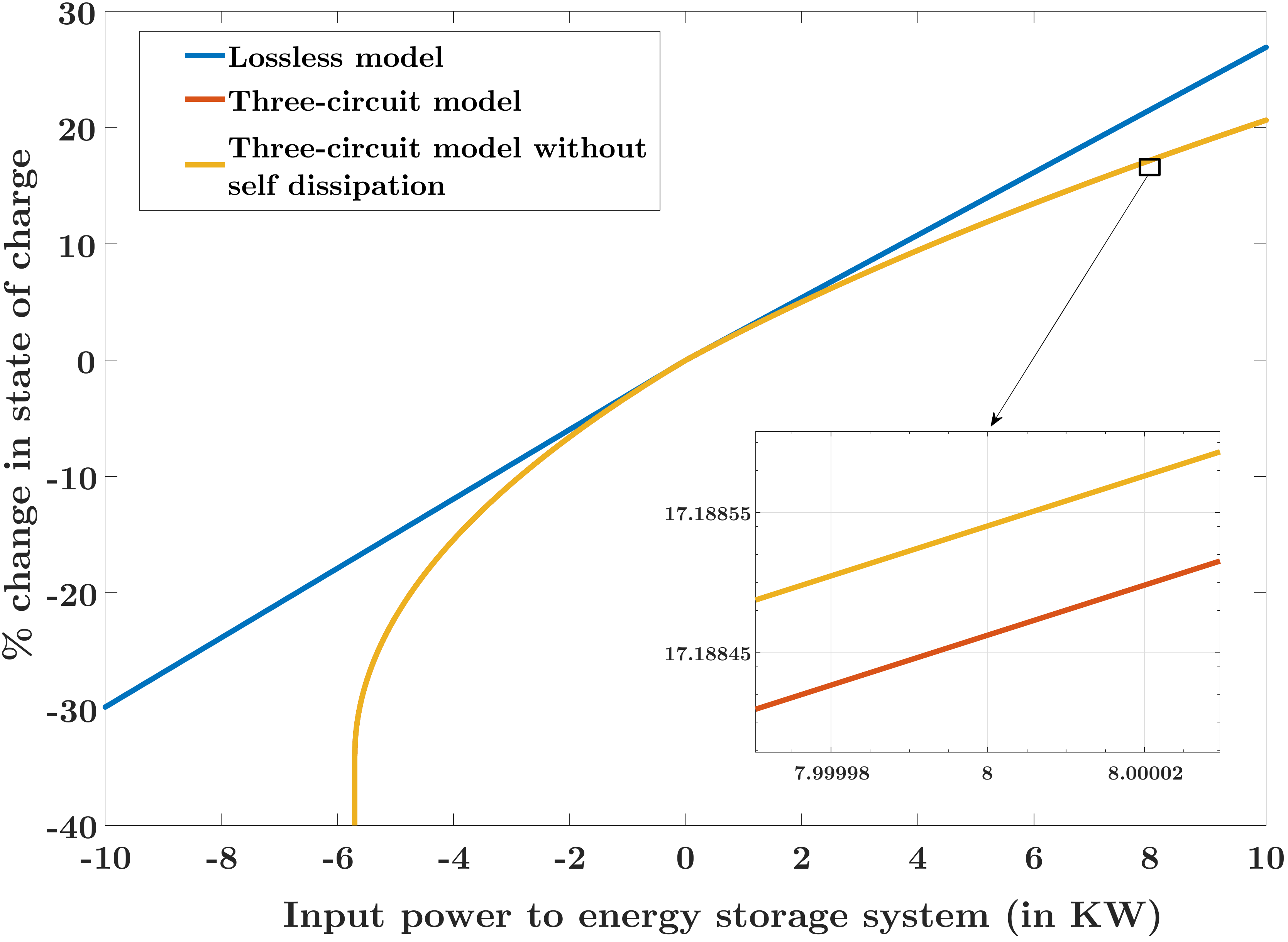}} 
        \end{minipage}
        \caption{Comparision of ESS models at 50\% SOC.}
        \label{fig:batmodelcomp}
    \end{figure}
    \section{Conclusion}
    \label{sec:conc}
    In this paper, we presented a privacy-preserving control scheme based on Bayesian risk and a three-circuit model to estimate the energy loss associated with a control action. The controller is modeled as a PO-MDP controlled sensor to maximize the Bayesian risk function of an adversarial hypothesis testing and the resulting nonlinear optimization objective is solved in a backward recursion. Extensive numerical experiments were carried out to evaluate the performance of the controller thoroughly. Especially, we tested the controller's performance against a state-of-the-art NILM algorithm using real energy consumption data. We investigated the effect of the initial state of the energy storage system on the controller's performance. An important conclusion from this work is that the privacy leakage can be reduced by using an energy storage system but at the expense of energy loss. Without an accurate model, the error in state estimation propagates and if not corrected, leads to suboptimal privacy control.
    
    Future work will focus on the trade-off between the privacy and energy loss, time dependency of the model parameters as well as control strategy and including more energy storage technologies. 
    \bibliographystyle{ieeetran}
    \bibliography{refs}
\end{document}